\documentclass[journal]{IEEEtran}
\usepackage{mathtools}
\usepackage{amsthm}
\usepackage{graphicx}
\usepackage{mathrsfs}
\usepackage{caption}
\usepackage[crop=preview]{pstool}
\usepackage{pgfplots}
\pgfplotsset{compat=newest}
\pgfplotsset{plot coordinates/math parser=false}                                                    
\usepackage{balance}
\usepackage{amssymb}
\usepackage{amsmath} 
\usepackage{subfigure}

\newcommand{\R}{\mathbb{R}}
\newcommand{\range}{\operatorname{rge}}
\newcommand{\gph}{\operatorname{gph}}
\newcommand{\ie}{{\itshape i.e.}}
\newcommand{\eg}{{\itshape e.g.}}
\newcommand{\Sone}{\mathbb{S}^1}

\newcommand{\dom}{\operatorname{dom}}
\newcommand{\sort}{\operatorname{sort}}
\newcommand{\dist}{\operatorname{d}}

\newcommand{\twopi}{2\pi}
\newcommand{\nex}{\operatorname{next}_{\mathcal{V}}}

\newcommand{\Hy}{\mathcal{H}}

\newcommand{\1}{{\bf 1}}

\newcommand{\dstar}{{2\pi\frac{n-1}{n}}}

\newtheorem{theorem}{Theorem}

\newtheorem{definition}{Definition}

\newtheorem{lemma}{Lemma}

\newtheorem{assumption}{Assumption}

\newtheorem{problem}{Problem}
\newtheorem{proposition}{Proposition}
\newtheorem{remark}{Remark}

\title{\LARGE \bf Robust Almost Global Splay State Stabilization of Pulse Coupled Oscillators}
\author{Francesco~Ferrante~\IEEEmembership{Member,~IEEE,} Yongqiang Wang~\IEEEmembership{Senior Member,~IEEE}
\thanks{Francesco Ferrante and Yongqiang Wang are with the Department of Electrical and Computer Engineering, Clemson University,
Clemson, SC 29631, 
Email: francesco.ferrante.2011@ieee.org, yongqiw@clemson.edu}
\thanks{This work has been partially supported by the Institute for Collaborative Biotechnologies
through grant W911NF-09-0001. A preliminary and abridged version of this paper has been presented in \cite{Ferrante2016}.}}
\begin{document}
\maketitle
\vspace{-0.5cm}

\begin{abstract}
This technical note deals with the problem of asymptotically stabilizing the splay state configuration of a network of identical pulse coupled oscillators through the design of the their phase response function. The network of pulse coupled oscillators is modeled as a hybrid system. The design of the phase response function is performed to achieve almost global asymptotic stability of a set wherein oscillators' phases are evenly distributed on the unit circle. To establish such a result, a novel Lyapunov function is proposed. Robustness with respect to frequency perturbation is assessed.
Finally, the effectiveness of the proposed methodology is shown in an example.
\end{abstract}
\vspace{-0.5cm}
\section{Introduction}
Pulse coupled oscillators are oscillators connected together to form a network of systems interacting through the exchange of impulsive messages. In particular, each time an oscillator completes a cycle it sends a pulse (fires) to the other oscillators. Such oscillators, in response to the received pulse, reset their phase according to a phase response function; see \cite{canavier2010pulse,mirollo1990synchronization}. Despite their simplicity, pulse coupled oscillators may manifest complex collective behaviors. For such a reason, their adoption as modeling and design tools in complex biological and engineered systems has provided interesting results in different fields; see, \eg, \cite{mirollo1990synchronization,pagliari2010,peskin1975mathematical} just to cite a few.
Concerning engineered systems, the most investigated behaviors in pulse coupled oscillators are synchronization and desynchronization. When a network of pulse coupled oscillators evolve in synchronization, oscillators evolve with the same phase each other. When a network of pulse coupled oscillators reaches desynchronization, oscillators have different phases and the time between consecutive firings is constant \cite{pagliari2010}. Desynchronization has relevant applications in engineered systems as it provides a simple, robust, and decentralized way to generate round-robin schedulers; \cite{pagliari2010}.  
In particular, as pointed out in \cite{pagliari2010}, a network of pulsed coupled oscillators can manifest two possible kinds of desynchrony: splay state (or strict desynchronization) and weak desynchronization. In the splay state, oscillators' phases are evenly distributed around the unit circle (\cite{sepulchre2007stabilization}) implying that the time between consecutive firing instants is constant. 
 Instead, in the weak desynchronization, there is no equal spacing between oscillators' phases and equal separation of firing times is achieved by  
making oscillators decrease their phase at each firing time. This kind of desynchronized behavior has been studied in \cite{philips2015} via hybrid systems tools.  It is worthwhile to observe that although the occurrence of both behaviors ensures constant firing intervals, since in the splay state configuration no interaction between oscillators is needed, such a configuration reveals to be more robust to pulse lossess and therefore preferable in some applications; see \cite{pagliari2010}. 
 
In this technical note, we address the problem of stabilizing the splay state configuration of a network of pulse coupled oscillators connected in an all-to-all topology. Inspired by \cite{nunez2015synchronization, nunez2016synchronization, philips2015}, we adopt a hybrid systems approach to tackle the considered problem. In particular, we model the network of pulse coupled oscillators as a hybrid system $\mathcal{H}$ in the framework of \cite{goebel2012hybrid}. Then, building on this model, similarly as in \cite{philips2015}, we recast the splay state stabilization problem as a set stabilization problem for $\mathcal{H}$. 
As a second step, by relying on the notion of shortest containing arc (\cite{dorfler2014synchronization,nunez2016synchronization}), we show that the splay state configuration coincides with the (unique) configuration for which the length of the shortest containing arc is maximized. Building on such a property, we define a Lyapunov function related to the length of the shortest containing arc. Finally, thanks to the invariance principle for hybrid systems \cite{sanfelice2007invariance}, inspired by \cite{pagliari2010}, we provide sufficient conditions on the phase response function to ensure almost global asymptotic stability of the splay state configuration. Furthermore, by relying on some regularity properties satisfied by the data of $\mathcal{H}$, we show 
that the splay state configuration is structurally robust to frequency perturbations, making the results presented in this paper appealing in practice. Finally, the theoretical findings are illustrated via a numerical example.
\subsection*{Contribution}
The contribution of this paper are as follows. Firstly, we formalize the splay state stabilization problem for pulse coupled oscillators in the hybrid systems framework in \cite{goebel2012hybrid}, which allows to obtain formal robustness guarantees about the perturbed behavior of the network. In this sense, this work can be seen as complementary to \cite{philips2015}, where weak desynchronization is considered.
Secondly,  we propose sufficient conditions on the phase response function ensuring asymptotic stability of the splay state configuration. Thirdly, for the first time, we make use of the notion of shortest containing arc in a desynchronization setting to build a Lyapunov function. Observe that, as opposed to synchronization where the relationship between the 
length of the shortest containing arc and the synchronous configuration is obvious, in the case of splay state, such a relationship is nontrivial and needs to be worked out.

A preliminary version of this work was presented in \cite{Ferrante2016}. On the one hand, this paper aims at completing the preliminary results presented in \cite{Ferrante2016} by providing further technical clarifications, refinements, and proofs. On the other hand, a thorough comparison with existing results as well as the analysis in the presence of frequency perturbations are presented.

The remainder of the technical note is organized as follows. Sections I.A concerns preliminary notions/results. 
Section I.B gives some basic notions on hybrid systems. 
Section II presents the system under consideration, the hybrid modeling and the problem we solve. Section III is dedicated to the main results. Section IV is devoted to a numerical example. Numerical simulations of hybrid systems are performed in \emph{Matlab} via the \emph{Hybrid Equations Toolbox} \cite{sanfelice2013toolbox}.

{\small {\bf Notation}:
The set $\mathbb{N}$ is the set of strictly positive integers, $\mathbb{N}_0=\mathbb{N}\cup\{0\}$, $\mathbb{R}_{\geq 0}$ represents the set of non-negative real scalars, and $\mathbb{C}$ represents the set of complex numbers. We denote $\iota=\sqrt{-1}$ as the imaginary unit and $\Sone=\{z\in\mathbb{C}\colon \vert z \vert=1\}$ the unit circle, where $\vert z\vert$ is the modulus of $z$. 
Given a set $S$, $\mathcal{C}(S)$ denotes the cardinality of $S$.
For a vector $x\in\mathbb{R}^n$, $x_i$ denotes its $i$-th entry and $\vert x\vert$ the Euclidean norm. The symbol $\1_n$ denotes the vector of $\R^n$ whose entries are equal to one. Given two vectors $x,y$, we denote $(x,y)=[x'\,\,y']'$. The set $\mathbb{B}$ is the closed unit ball, of suitable dimensions, in the Euclidean space. 
Given a vector $x\in\mathbb{R}^{n}$  and a closed set $\mathcal{A}$, the distance of $x$ to $\mathcal{A}$ is defined as 
$\vert x \vert_{\mathcal{A}}=\inf_{y\in {\mathcal{A}}} \vert x-y \vert$. Given $\mathcal{I}\subset [0,\twopi]$, we denote $\Psi(\mathcal{I})=\{z\in \Sone\colon z=e^{i\omega}\quad \omega\in\mathcal{I}\}$. Given a set $S$, $m(S)$ stands for the Lebesgue measure of $S$.
Let $n\in\mathbb{N}$, we denote $\mathcal{V}=\{1,2,\dots,n\}$. Given $i\in\mathcal{V}$, we denote $\nex(i)=i+1$ if $i<n$ while $\nex(i)=1$ if $i=n$. Given a vector $x\in\R^n$, we denote $\sort(x)$ the vector whose entries correspond to the entries of $x$ sorted from the smallest to the largest value. Given a function $f\colon \R^n\rightarrow \R^m$, we denote $\range f$ the range of $f$. Let 
$F\colon\R^n\rightrightarrows\R^m$ be a set-valued mapping, we denote $\gph F=\{(x,y)\in\R^{n}\times \R^m\colon y\in F(x)\}$. Given a function $f\colon\dom f\rightarrow\R$ and a real scalar $\mu$, $f^{-1}=\{x\in\dom f\colon f(x)=\mu\}$ is the $\mu$-level set of $f$.}
\vspace{-0.1cm}
 
\subsection{Preliminaries}
\begin{definition}
Given $\theta_1,\theta_2\in\R$, $\dist(\theta_1,\theta_2)$
represents the geodesic distance between  the point $x_1=e^{\iota \theta_1}$ and $x_2=e^{\iota \theta_2}$ on the unit circle $\Sone$.
\end{definition}
In the sequel, with a slight abuse of notation, given $\theta_1,\theta_2\in\R$, we refer to $\dist(\theta_1,\theta_2)$ as the geodesic distance between $\theta_1$ and $\theta_2$.
\begin{definition}
A subset $\alpha\subset\Sone$ is said to be an arc if it is closed and connected.
Given $\theta\in\R^n$ and an arc $\alpha\subset\Sone$. We say that $\alpha$ contains $\theta$ if for each $i=1,2,\dots,n$, $x_i=e^{\iota \theta_i}\in \alpha$.
In particular, given $\theta\in\R^n$, we denote $\Gamma(\theta)\subset \Sone$ the set of all the arcs containing $\theta$. Given an arc $\alpha\subset\Sone$, we denote $\ell(\alpha)\in[0,2\pi)$ the length of the arc $\alpha$.
\end{definition}
\begin{definition}
Given $\theta\in\R^n$, we denote $\gamma(\theta)$ the length of the shortest arc containing
$\theta$. In particular\footnote{It can be readily shown that for any $\theta\in\R^n$, $\mathcal{C}(\Gamma(\theta))=n$. Therefore, $\gamma$ is well defined.}, $\gamma(\theta)=\min_{\omega\in\Gamma(\theta)} \ell(\omega)$. 
\end{definition}
Now, consider the following result which will be useful in the sequel.
\begin{lemma}
\label{lemma:length}
For each $\theta\in\R^n$, one has $\gamma(\theta)\leq 2\pi-\frac{2\pi}{n}$.
\end{lemma}
\begin{proof}
Without loss of generality, assume by contradiction that there exists $\hat{\theta}\in[0,2\pi]^n$ with $0=\hat{\theta}_1\leq \hat{\theta}_2\leq \dots\leq \hat{\theta}_n$, such that 
$\gamma(\hat{\theta})>2\pi-\frac{2\pi}{n}$. Since $2\pi-\frac{2\pi}{n}\geq \pi$, 
as argued in \cite{nunez2015synchronization}, one has that $\gamma(\hat{\theta})=2\pi-\underset{i\in\mathcal{V}}{\max} \dist(\hat{\theta}_i,\hat{\theta}_{\nex(i)})$, hence for each $i\in\mathcal{V}$, $\dist(\hat{\theta}_i,\hat{\theta}_{\nex(i)})<\frac{2\pi}{n}$. On the other hand, since by assumption $\gamma(\hat{\theta})> 2\pi-\frac{2\pi}{n}\geq \pi$, it can be easily noticed that necessarily for each $i\in\mathcal{V}\setminus\{n\}$, $\dist(\hat{\theta}_i,\hat{\theta}_{\nex(i)})=\vert \hat{\theta}_i-\hat{\theta}_{\nex(i)}\vert=\hat{\theta}_{\nex(i)}-\hat{\theta}_i$ and $\dist(\hat{\theta}_n,\hat{\theta}_{1})=2\pi-\hat{\theta}_n$.
 Thus, 
$\hat{\theta}_{n}=\sum_{i\in\mathcal{V}\setminus\{n\}}(\hat{\theta}_{\nex(i)}-\hat{\theta}_i)<2\pi\frac{n-1}{n}$
but this implies that $\dist(\hat{\theta}_n,\hat{\theta}_1)=2\pi-\hat{\theta}_n>\frac{2\pi}{n}$ contradicting the fact that  for each $i\in\mathcal{V}$, $\dist(\hat{\theta}_i,\hat{\theta}_{\nex(i)})<\frac{2\pi}{n}$.
\end{proof}
\subsection{Hybrid systems}
\label{sec:hybrid}
We consider hybrid systems with state $x\in\R^n$ of the form 
$$
\mathcal{H}\left\lbrace
\begin{array}{ccll}
\dot{x}&=&f(x)&\quad x\in C\\
x^+&\in&G(x)&\quad x\in D
\end{array}\right.
$$
that we represent by the shorthand notation $\mathcal{H}=(C, f, D, G)$. A solution to $\mathcal{H}$ is any hybrid arc defined over a subset of $\R_{\geq 0}\times\mathbb{N}_0$ that satisfies the dynamics of $\mathcal{H}$.
A solution to a hybrid system is said to be \emph{complete} if its domain is unbounded and \emph{maximal} if it is not the truncation of another solution. Given a set $S$, we denote $\mathcal{S}_{\mathcal{H}}(S)$ the set of all maximal solutions $\phi$ to $\mathcal{H}$ with $\phi(0,0)\in S$.
Given $\mathcal{H}=(C,f,D,G)$, we say that $\mathcal{H}$ satisfies the \emph{hybrid basic conditions} (\cite{goebel2012hybrid}) if: $C$ and $D$ are closed in $\R^n$; $f\colon \R^n\rightarrow\R^n$ is continuous on $C$, $G\colon\R^n\rightrightarrows\R^n$ is an outer semicontinuous set-valued map \footnote{A set-valued map $M\colon \R^n\rightrightarrows\R^n$ is outer semicontinuous if its graph is closed; see \cite{rockafellar2009variational}.} nonempty and locally bounded on $D$. 
\begin{definition}[$(\tau,\varepsilon)$-closeness of hybrid arcs \cite{goebel2012hybrid}]
Given $\tau\geq 0$ and $\varepsilon>0$, two hybrid arcs $\xi_1$ and $\xi_2$ are $(\tau,\varepsilon)$-close if
\begin{itemize}
\item[(a)] for all $(t,j)\in\dom \xi_1$ with $t+j\leq \tau$ there exists $s$ such that $(s,j)\in\dom\xi_2$, $\vert t-s\vert<\varepsilon$, and $\vert \xi_1(t,j)-\xi_2(s,j)\vert\!<\!\varepsilon$;
\item[(b)] for all $(t,j)\in\dom \xi_2$ with $t+j\leq \tau$ there exists $s$ such that $(s,j)\in\dom\xi_1$, $\vert t-s\vert<\varepsilon$, and 
$\vert \xi_2(t,j)-\xi_1(s,j)\vert\!<\!\varepsilon$.
\end{itemize}
\end{definition}
\section{Problem statement}
\subsection{System description}
In this paper, we analyze an all-to-all network of $n\in\mathbb{N}\setminus\{1\}$ identical pulse-coupled oscillators (\emph{PCO}). 
Each oscillator is characterized by a phase variable $x_i\in[0,2\pi]$ for each $i\in\mathcal{V}\coloneqq \{1,2,\dots,n\}$, that evolves continuously from $0$ to $2\pi$ according to integrate-and-fire dynamics, \ie, $\dot{x}_i=\omega$, where $\omega\in\R_{>0}$ is the natural frequency of the oscillators;
when, $x_i$ reaches $2\pi$, the oscillator $i$ fires, \ie, emits a pulse and resets its phase to zero. Such a pulse is instantaneously broadcast through the network to the other oscillators $x_j$ with $j\in\mathcal{V}\setminus\{i\}$ that in turn reset their phases according to a phase response function $Q$ \cite{canavier2010pulse}, \ie, $x_j^+=x_j+Q(x_j)$, where $x^+_j$ denotes the value of $x_j$ right after a reset.

The problem studied in this technical note consists of designing the function $Q$ in a way such that the configuration in which the $n$ phases are evenly distributed on the unit circle (\emph{splay state} \cite{sepulchre2007stabilization}) is asymptotically stable. In the sequel, we refer to this problem as \emph{splay state stabilization problem}.
\begin{remark}
Assuming an all-to-all communiction between oscillators allows to provide a simple yet robust solution to the considered problem, which is therefore suitable for practical implementations. This assumption is common in the literature on desynchronization of PCOs; see, \eg, \cite{pagliari2010, philips2015}. Indeed, although this assumption simplifies the analysis of the resulting network, it is often used in engineered systems, like wireless network \cite{pagliari2010} and in the study of some biological systems such as cardiac pacemaker cells \cite{peskin1975mathematical} and populations of spiking neurons \cite{gerstner2002spiking}. Further insights about the complications encountered in the case of more general topologies are illustrated throughout the paper.
\end{remark}
\subsection{Hybrid modeling and dynamical properties}
Due to the hybrid behavior of a network of PCOs, as in \cite{nunez2015synchronization,philips2015}, in this paper we model such a network as a hybrid system $\mathcal{H}$ with state $x=(x_1,x_2,\dots, x_n)\in\R^n$, following the formalism in \cite{goebel2012hybrid}. To this end, we need to define $C,f,D,G$ such that $\mathcal{H}=(C,f,D,G)$ provides a model of the PCOs network dynamics.

In particular, we define the flow map and the flow set as follows
\begin{equation}
\label{eq:flowdata}
\begin{array}{lr}
C\coloneqq [0,2\pi]^n,&
f(x)\coloneqq \1_n\omega \qquad\forall x\in C
\end{array}
\end{equation}
 While the jump set can be defined as follows $D=\{x\in[0,2\pi]^n\colon \underset{i\in\mathcal{V}}\max\, x_i=2\pi\}$
so to enforce a jump whenever at one oscillator fires. Finally, to define the jump map $G$ of $\mathcal{H}$,
we follow the idea in \cite{philips2015}. In particular, for each $x\in D$ we set 
$
G(x)=(g_1(x),g_2(x),\dots, g_n(x))
$
where for each $i\in\mathcal{V}$
\begin{equation}
\label{eq:gi}
g_i(x)\hspace*{-0.1cm}\coloneqq\hspace*{-0.1cm}\left\lbrace\begin{array}{ll}
\hspace*{-0.25cm}0& \hspace*{-0.3cm}\mbox{if}\, x_i=2\pi, \forall j\in\mathcal{V}\setminus\{i\}\, x_j\neq 2\pi\\
\hspace*{-0.25cm}x_i+Q(x_i)&\hspace*{-0.3cm} \mbox{if}\, x_i \neq 2\pi, \exists j\in\mathcal{V}\setminus\{i\} \,x_j=2\pi\\
\hspace*{-0.25cm}\{0,x_i+Q(x_i)\}&\hspace*{-0.3cm} \mbox{if}\, x_i=2\pi, \exists j\in\mathcal{V}\setminus\{i\}\, x_j=2\pi
\end{array}\right.
\end{equation}
Notice that whenever one single oscillator fires, its state is reset to zero while the phase of each other oscillator is reset according to the function $Q$. If multiple oscillators fire, then the jump map is set-valued allowing to obtain a hybrid system satisfying the hybrid basic conditions. Moreover, to make $\mathcal{H}$ an accurate description of the evolution of the network of PCOs previously described, we consider the following standing assumption; see \cite{wang2012}.
\begin{figure}[h]
\begin{center}
\pstool[scale=1, trim=15mm 4mm  15mm 7mm, clip=true]{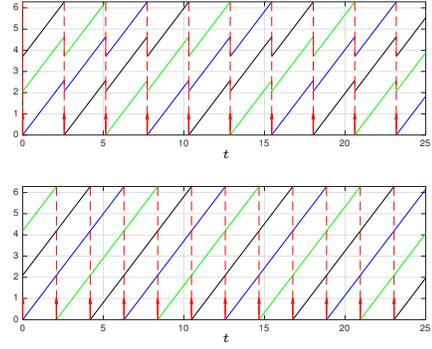}
{\psfrag{t}{$t$}}
\end{center}
\caption{The evolution of a network of three PCOs with $\omega=1$ in the: desynchronized behavior studied in \cite{philips2015} (top), splay state (bottom). Red arrows indicate the firing times. In both cases, the time in between firing events is constant. However, in the case of the (weak) desynchronization analyzed in \cite{philips2015}, equal firing intervals are maintained by pulse-caused phase shift, making this configuration sensitive to pulse losses.}
\label{fig:Sean}
\end{figure}
\begin{assumption}
\label{as:Standing}
The function $Q\colon [0,2\pi]\rightarrow\R$ is such that $\range (z\mapsto z+Q(z))\subset[0,2\pi]$
\hfill $\triangle$
\end{assumption}
The above assumption ensures that $G(D)\subset C\cup D=C$ avoiding the existence of solutions ending in finite time due to jumping outside $C\cup D$.
\begin{remark}
Although the dynamics of the hybrid system $\mathcal{H}$ and the problem setting may look similar to the ones studied in \cite{philips2015}, the problem addressed in this paper cannot be solved by employing the methodology proposed in \cite{philips2015}. Indeed, the desyncronized behavior studied in \cite{philips2015} (weak desynchronization) is defined as the behavior in which the separation between all of the firing events is equal (and nonzero). Obviously, the occurrence of this behavior does not ensure even spacing between phases; see the top figure in \figurename~\ref{fig:Sean}. In fact, the reset rule adopted in \cite{philips2015} is such that the firing of a node makes all the other oscillators decrease their phase, hence preventing from achieving the splay state configuration\footnote{The simulation shown in the top of  \figurename~\ref{fig:Sean} is performed by taking the parameter $\varepsilon$ in \cite{philips2015} equal to $-0.2$ and by picking an initial condition in the desynchronizaiton set given in \cite{philips2015}.}; see \figurename~\ref{fig:Sean}. Therefore, further developments are needed to solve the splay state stabilization problem analyzed in this paper.
\end{remark}
\begin{remark}
In real-time decentralized implementations, whenever multiple oscillators fire (usually) each of the firing oscillators resets its phase to zero overlooking the incoming pulse stimulus received from other firing oscillators. This behavior is well known in the literature of PCOs and prevents to achieve convergence towards the splay state for every initial condition; see, \eg, \cite{pagliari2010}.
It is interesting to observe that, although in the case of simultaneous firings our model gives rise to multiple solutions, the behavior described here above is still captured by our model. 
Indeed, if two or more oscillators fire simultaneously the definition of $G$ allows to consider (among the others) the solutions to $\mathcal{H}$ for which after one jump the phase of two or more oscillators are all equal to zero, preventing $\mathcal{H}$ from approaching the splay state configuration.
\end{remark}
Now consider the following assumption.
\begin{assumption}
\label{as:As1}
The function $Q\colon [0,2\pi]\rightarrow \R$ satisfies the following properties:
\begin{itemize}
\item[(A.1)] $Q$ is continuous on $[0,2\pi]$ 
\item[(A.2)] $Q$ is such that $Q(2\pi)\neq 0$ \hfill $\triangle$
\end{itemize}
\end{assumption}
\begin{proposition}
\label{prop:Prop1}
Let Assumption~\ref{as:As1} hold. Then, the following properties hold:
\begin{itemize}
\item[(a)] $\mathcal{H}$ satisfies the hybrid basic conditions; 
\item[(b)] For each initial condition $\xi\in C\cup D$, there exists at least one nontrivial solution to $\mathcal{H}$. In particular, every $\phi\in\mathcal{S}_{\mathcal{H}}(C\cup D)$ is complete, non Zeno, and $\sup_j\dom\phi=\infty$.
\end{itemize}
\end{proposition}
\begin{proof}
Point (a) is straightforward to prove. Indeed, $C$ and $D$ are closed in $\R^n$ and $f$ is continuous on $C$. Moreover, the jump map $G$ is obviously bounded and nonempty on $D$ and, similarly as shown in \cite{philips2015}, is outer semicontinuous on $D$ since $Q$ is continuous\footnote{This latter property can be proven by directly using the definition of outer semicontinuity \cite[Definition 5.9]{goebel2012hybrid}.}. Hence, (a) holds.
To prove (b), observe that since $\mathcal{H}$ satisfies the hybrid basic conditions,  existence of nontrivial solutions from $C\cup D$  and completeness of each $\phi\in\mathcal{S}_{\mathcal{H}}(C\cup D)$ follows directly from the application of \footnote{In particular, it can be verified that item (VC) in \cite[Proposition 6.10]{goebel2012hybrid} holds and that due to   $G(D)\subset C\cup D$ and $C\cup D$ compact,  items (b) and (c) in \cite[Proposition 6.10]{goebel2012hybrid} are ruled out.}\cite[Proposition 6.10]{goebel2012hybrid}. Moreover,  $\sup_j\dom\phi=\infty$ since $\phi$ is complete and the length of each flow interval is at most $\frac{2\pi}{\omega}$. 
To prove that each $\phi\in \mathcal{S}_{\mathcal{H}}(C\cup D)$ is non Zeno, we make use of \cite[Lemma 2.7]{sanfelice2007invariance}. In particular, first observe that due to (A.2), since $Q(2\pi)\neq 0$, one has $G(D)\cap D=\emptyset$. Pick any $\phi\in \mathcal{S}_{\mathcal{H}}$ and let $E=\bigcup_{j\in\mathbb{N}_0}[t_j, t_{j+1}]\times \{j\}$ be the domain of $\phi$. Then, since every maximal solution to $\mathcal{H}$ is precompact\footnote{A solution $\phi$ is precompact if it is bounded and complete.} and, as shown earlier, $\mathcal{H}$ satisfies the hybrid basic conditions, thanks to \cite[Lemma 2.7]{sanfelice2007invariance}, it follows that there exists a positive real scalar $\gamma$ (dependent on the solution $\phi$) such that for each $j\geq 1$, $t_{j+1}-t_j\geq \gamma$ and this rules out the existence of Zeno solutions, concluding the proof.
\end{proof}
\begin{remark}
The fact that $\mathcal{H}$ satisfies the hybrid basic conditions entails two main advantages. The first advantage is that working with a hybrid system fulfilling the hybrid basic conditions allows to exploit a large numbers of results available in the literature of hybrid systems such as the invariance principle in \cite{sanfelice2007invariance}. The second advantage is that the satisfaction of the hybrid basic conditions ensures that $\mathcal{H}$ is structurally robust to small perturbations, this aspect will be illustrated in Section III.b.
\end{remark}
\subsection{Splay state stabilization}
In this section, we formalize the splay state stabilization problem as a set stabilization problem for $\mathcal{H}$.
In this paper, we consider the following notions of stability for the hybrid system $\mathcal{H}$:
\begin{definition}(\cite[Definition 7.1]{goebel2012hybrid})
Let ${\cal A}\subset\mathbb{R}^n$ be a compact set.  The set ${\cal A}$ is 
\begin{itemize}
\item stable for $\mathcal{H}$ if for every $\epsilon>0$ there exists $\delta>0$ such that every solution to $\mathcal{H}$ with $\vert\phi(0,0) \vert_{\mathcal{A}}\leq \delta$ satisfies $\vert\phi(t,j) \vert_{\mathcal{A}}\leq \epsilon$ for all $(t,j)\in\dom\phi$;
\item locally attractive for $\mathcal{H}$ if every maximal solution to $\mathcal{H}$ is complete and 
 there exists $\mu\!>\!0$ such that for every maximal solution $\phi$ to $\mathcal{H}$ with $\vert\phi(0,0)\vert_{\mathcal{A}}\leq \mu$ one has $\lim_{t+j\rightarrow\!\infty\!}\vert\phi(t,j)\vert_{\mathcal{A}}\!=\!0$
\item locally asymptotically stable (LAS) for $\mathcal{H}$, if it is both stable and locally attractive for $\mathcal{H}$.
\end{itemize}
\end{definition} 
\begin{definition}(\cite[Definition 7.3]{goebel2012hybrid})
Let ${\cal A}\subset\mathbb{R}^n$ be a compact set locally asymptotically stable for $\mathcal{H}$. The basin of attraction of $\mathcal{A}$, denoted by $\mathcal{B}_{\mathcal{A}}$, is the set of points such that every $\phi\in\mathcal{S}_{\mathcal{H}}(\mathcal{B}_{\mathcal{A}})$ is bounded, complete,  and $\lim_{t+j\rightarrow \infty}\vert\phi(t,j) \vert_{\mathcal{A}}=0$.
\end{definition} 
\begin{remark}
Observe that the basin of attraction $\mathcal{B}_{\mathcal{A}}$ always contains points in $\R^n \setminus (C\cup D)$ since no solution exists from such points, implying that the completeness, boundedness, and convergence requirements are vacuously verified. 
\end{remark}
\begin{definition}(Almost global asymptotic stability)
Let ${\cal A}\subset\mathbb{R}^n$ be a compact set asymptotically stable for $\mathcal{H}$. The set $\mathcal{A}$ is said to be almost globally asymptotically stable for $\mathcal{A}$ if $m(\R^n\setminus\mathcal{B}_{\mathcal{A}})=0$. 
\end{definition}
\begin{remark}
Almost global asymptotic stability ensures that convergence towards the attractor $\mathcal{A}$ occurs for almost all initial conditions. Due to its relevance in practical applications, this property has been extensively used in the study of synchronization problems; see \cite{dorfler2014synchronization}.
\end{remark}
To make use of the above stability definition, analogously to \cite{philips2015},  we show that the splay state configuration can be captured by defining an adequate compact set $\mathcal{A}\subset C$ such that whenever the state of $\mathcal{H}$ belongs to $\mathcal{A}$ then, the system is in the splay state configuration. To this aim, we first consider the following definition:
\begin{definition}(Splay state)
Let $\phi$ be a solution to $\mathcal{H}$ and let for each $(t,j)\in \dom\phi$, $\psi(t,j)\coloneqq \sort(\phi(t,j))$.
We say that  $\phi$ is a splay state solution if there exists a strictly positive constant $c$, such that for each $(t,j)\in\dom\phi$ and each $i\in\mathcal{V}$, 
one has $\dist(\psi_i(t,j), \psi_{\nex(i)}(t,j))=c$.
\end{definition}
In particular, it can be easily argued that the constant $c$ in the above definition has to be equal to $\frac{2\pi}{n}$ to guarantee even spacing between the phases.    
This observation leads to the following definition of the set $\mathcal{A}$ characterizing the splay state configuration.
\begin{equation}
\label{eq:SetA}
\mathcal{A}=\left\{x\in C\colon y=\sort(x),\dist(y_i,y_{\nex(i)})\!=\!\frac{2\pi}{n}\,\,\,\forall i\in\mathcal{V}\right\}
\end{equation}
Notice that $\mathcal{A}$ is compact since it is bounded (included in $C$ that is bounded) and closed since the function $\sort\colon\R^n\rightarrow \R^n$ is continuous\footnote{Continuity of the mapping $x\mapsto \sort(x)$ can be easily shown by noticing that for each $x\in\R^n$ and any sequence $x_k\rightarrow x$, one has that $\sort(x_k)\rightarrow \sort(x)$.}.
\begin{remark}
\label{rem:line}
Clearly, due to the difference between the splay state and the desynchronized behavior studied in \cite{philips2015}, the set $\mathcal{A}$ in \eqref{eq:SetA} is different from the desynchronization set defined in \cite{philips2015}. However,  similarly to \cite{philips2015}, the set $\mathcal{A}$ in \eqref{eq:SetA} also enjoys a simple geometric characterization. In particular, it can be easily shown that the set $\mathcal{A}$ corresponds to the union of $n!$ segments\footnote{The fact that the set $\mathcal{A}$ can be represented as the union of $n!$ segments can be easily seen from the definition of the set $\mathcal{A}$ by observing that: for each $i\in\mathcal{V}$, the constraint on the geodesic distance $\dist(\cdot)$ given in \eqref{eq:SetA} gives rise to a segment in the space $[0,\twopi]^n$; the $n$ constraints considered in \eqref{eq:SetA} can be described equivalently by considering all the orderings of the components of the vector $x$, which are $n!$. Thus, to determine such segments, for each of the considered orderings, it suffices to pick two points belonging to the set $\mathcal{A}$.} in $[0,\twopi]^n$, which in turn can be written as the intersection of a set of lines $\mathcal{L}$ in $\R^n$ with the box $[0,\twopi]^n$. 
Nonetheless, in contrast to \cite{philips2015}, in this paper this characterization is not useful to derive sufficient conditions for asymptotic stability of the set $\mathcal{A}$. This aspect is shown in Section III.
\end{remark}
As already mentioned earlier, one cannot expect to make the set $\mathcal{A}$ attractive for every initial condition $\xi\in C$. Indeed, let 
\begin{equation}
\label{eq:calX}
\mathcal{X}=\{x\in C\colon \exists i, j\in\mathcal{V}\,\,\mbox{s.t.}\,\,i\neq j ,\,\mbox{and}\,\, \dist(x_i,x_j)=0\}
\end{equation}
Then, for each $\xi\in\mathcal{X}$ and for any choice of the function $Q$, there exists a solution $\phi$ to $\mathcal{H}$ with $\phi(0,0)=\xi$ which has identically synchronized components, \ie, $\phi$ does not approach $\mathcal{A}$. On other other hand, since the set $\mathcal{X}$ has measure zero, 
we solve the following problem:
\begin{problem}
\label{prob:Main}
Design the function $Q$ such that the set $\mathcal{A}$ defined in \eqref{eq:SetA} is almost globally asymptotically stable for the hybrid system $\mathcal{H}$.
\end{problem}
\section{Main Results}
\subsection{Nominal behavior}
This first result relates the set $\mathcal{A}$ to another compact set whose role is key in the derivation of the main result presented in this paper.
\begin{proposition}
\label{prop:Main}
Let $\mathcal{A}_c=\left\{x\in[0,2\pi]^n\colon \gamma(x)=2\pi\frac{n-1}{n}\right\}$, then $\mathcal{A}=\mathcal{A}_c.$
\end{proposition}
\begin{proof}
First observe that for $n=2$ the result is trivial, so in the sequel we assume $n>2$.
We prove the above claim in two parts. First we show that $\mathcal{A}\subset\mathcal{A}_c$. 
The above set inclusion is trivial to show. In particular, pick $x=(x_1,x_2,\dots,x_n)\in\mathcal{A}$ and assume without loss of generality that $0=x_1<x_2<\dots <x_n$. Since for each $i\in\mathcal{V}$, 
$\dist(x_i,x_{\nex(i)})=\frac{2\pi}{n}$, one has $x_2=\frac{2\pi}{n}, x_3=\frac{4\pi}{n},\dots, x_n=2\pi\frac{n-1}{n}$. Hence, it can be easily argued that $\gamma(x)=\sum_{i\in\mathcal{V}\setminus\{n\}}\dist(x_i,x_{\nex(i)})=2\pi\frac{n-1}{n}$, showing that $\mathcal{A}\subset\mathcal{A}_c$.

 Now we show that $\mathcal{A}_c\subset\mathcal{A}$. Pick $x\in\mathcal{A}_c$, then $\gamma(x)=2\pi\frac{n-1}{n}>\pi$. Without loss of generality, assume $0=x_1\leq x_2\leq \dots \leq x_n$. Since $\gamma(x)\geq \pi$, as shown in \cite{nunez2015synchronization}, one has $\gamma(x)=2\pi-\underset{i\in\mathcal{V}}{\max} \dist(x_i,x_{\nex(i)})$, which, due to $x\in\mathcal{A}_c$, implies 
\begin{equation}
\label{eq:MaxD}
\underset{i\in\mathcal{V}}{\max} \dist(x_i,x_{\nex(i)})=\frac{2\pi}{n}
\end{equation}
Suppose by contradiction that $x\notin\mathcal{A}$. Then, there exists $i\in\mathcal{V}$ such that $\dist(x_i,x_{\nex(i)})\neq \frac{2\pi}{n}$. Assume without loss of generality $i=1$. Thanks to \eqref{eq:MaxD}, it has to be $\dist(x_1,x_2)<\frac{2\pi}{n}$, which yields $x_2\in[0,\frac{2\pi}{n})$. Therefore, since $\gamma(x)>\pi$, by following the same steps as in the proof of Lemma~\ref{lemma:length}, one can easily show that necessarily $x_n\in(\pi, 2\pi\frac{n-1}{n})$, \ie, $\dist(x_n,x_1)>\frac{2\pi}{n}$ contradicting \eqref{eq:MaxD}.
\end{proof}
The above result shows that the set $\mathcal{A}$ corresponds to the set in which the length of the shortest containing arc is maximized. This fact is crucial and allows to exploit the notion of containing arc, in a dual fashion with respect to \cite{nunez2016synchronization}, for the construction of a Lyapunov function candidate for $\mathcal{H}$.

Now, with the aim of solving Problem~\ref{prob:Main}, inspired by \cite{pagliari2010}, we consider the following assumption on the function $Q$.
\begin{assumption}
\label{as:AsQ}
The function $Q$  is such that,  the function $z\mapsto z+Q(z)$ is injective and $\gph Q\subset \mathcal{Q}_1\cup \mathcal{Q}_2$
where 
$$
\begin{aligned}
&\mathcal{Q}_1\hspace{-0.1cm}\coloneqq\hspace{-0.1cm}\left\{(w_1,w_2)\in\R^2\colon w_1\in \left[0,\frac{2\pi(n-1)}{n}\right] \land w_2=0\right\}\\
&\begin{split}\mathcal{Q}_2&\coloneqq\left\{(w_1,w_2)\in\R^2\colon\right. \\
&\left. w_1\in \left(\frac{2\pi(n-1)}{n}, 2\pi\right]\land\frac{2\pi(n-1)}{n}-w_1<w_2<0\right\}\end{split}
\end{aligned}
$$\hfill $\triangle$
\end{assumption}
\begin{remark}
Furthermore, notice that the satisfaction of Assumption~\ref{as:AsQ} guarantees the satisfaction of Assumption~\ref{as:Standing}, so the latter does not need to be checked.
\end{remark}
Essentially, the above assumption guarantees that at each firing, the oscillators whose phase belongs to $(\dstar,2\pi]$ decrease their phase without exiting $(\dstar,2\pi]$. Instead, the oscillators whose phase belongs to $[0,\dstar]$ are unaffected by the firing event. 
Moreover, having assumed $z\mapsto z+Q(z)$ to be injective guarantees that the firing of one oscillator does not make the phase of two or more oscillators jump to the same value. Indeed, in this case convergence towards $\mathcal{A}$ is impossible.   
This aspect and other important properties inherited by $\mathcal{H}$ from Assumption~\ref{as:As1} and Assumption~\ref{as:AsQ} are stated in the following result.
 \begin{lemma}
 \label{lemma:X}.
Let $\mathcal{X}$ be defined as in \eqref{eq:calX} and let $\mathcal{U}=C\setminus \mathcal{X}$. If Assumption~\ref{as:As1} and Assumption~\ref{as:AsQ} hold, then for every $\phi\in \mathcal{S}_{\mathcal{H}}(\mathcal{U})$ one has $\range\phi\subset \mathcal{U}$.
\end{lemma}
\begin{proof}
Due to the definition of the flow dynamics, solution cannot exit the set $\mathcal{U}$ by flowing. Therefore, one needs to rule out the existence of solutions entering the set $\mathcal{X}$ via a jump. To this end, pick a solution $\phi\in\mathcal{S}_{\mathcal{H}}(\mathcal{U})$ and denote $(t,k)\mapsto \phi(t,k)\coloneqq(\phi_1(t,k), \phi_2(t,k),\dots, \phi_n(t,k))$. Assume that $(t,k)\in\dom \phi$ is such that $(t,k+1)\in\dom\phi$ and that $\phi(t,k)\in\mathcal{U}$. Since, $\phi(t,k)\in\mathcal{U}$, it follows that for each $i\neq j \in\mathcal{V}$ one has $\dist (\phi_i(t,k),\phi_j(t,k))>0$. Let $\bar{i}\in\mathcal{V}$ such that $\phi_{\bar{i}}(t,k)=2\pi$, then $\phi_{\bar{i}}(t,k+1)=0$. While, for each $i\neq \bar{i}\in\mathcal{V}$, $\phi_i(t,k+1)=\phi_i(t,k)+Q(\phi_i(t,k))$. Thus, since $\phi(t,k)\in\mathcal{U}$ and $z\mapsto z+Q(z)$ is injective, for each $i\neq j\in\mathcal{V}\setminus\{\bar{i}\}$ $\phi_i(t,k+1)\neq \phi_j(t,k+1)$. To conclude the proof, one needs to show that for each $i\in\mathcal{V}\setminus\{\bar{i}\}$ $\phi_i(t,k+1)\notin\{0,\twopi\}$. But this is verified. Indeed, since $\phi(t,j)\in\mathcal{U}$, either $\phi_i(t,k)\in(0,\dstar]$ or $\phi_i(t,k)\in(\dstar, 2\pi)$. In the first case, $\phi_i(t,k+1)=\phi_i(t,k)\notin\{0,\twopi\}$, while in the second case
$\phi_i(t,k+1)\in[\phi_i(t,k)-\frac{2\pi}{n}, \phi_i(t,k))\subset (\dstar, 2\pi)$, concluding the proof.
\end{proof}
\begin{remark}
It is worthwhile to note that the proof of the above claim strongly relies on the fact that
the oscillators are connected through an all-to-all network.
Indeed, when such an assumption cannot be fulfilled, even if the function $z\mapsto z+Q(z)$ is injective, maximal solutions to $\Hy$ may leave the set $\mathcal{U}$ in finite time preventing from approaching $\mathcal{A}$. This fact makes the solution of the splay state stabilization problem for general topologies very challenging and unachievable with the simple interaction mechanism considered in this paper. Such an aspect is acknowledged by the existing literature on PCOs such as \cite{degesys2008towards}.  
\end{remark}
Now we give the main result of this paper.
\begin{theorem}
\label{theorem:Main}
Let Assumption~\ref{as:As1} and Assumption~\ref{as:AsQ} hold. Then, the set $\mathcal{A}$ defined in \eqref{eq:SetA} is 
almost globally asymptotically stable for $\mathcal{H}$ and the basin of attraction of $\mathcal{A}$ includes $\mathcal{U}$.
\end{theorem}
\begin{proof}
For every $x\in \R^n$, define the following Lyapunov function candidate $V(x)=\dstar-\gamma(x)$.
Observe that $V$ is continuous in $\R^n$ (see \cite{Ferrante2016}) and thanks to Lemma~\ref{lemma:length} and Proposition~\ref{prop:Main} positive definite with respect to $\mathcal{A}$ on $C\cup D$. Define, $u_C(x)=0$ for each $x\in C$ and $u_C(x)=-\infty$ otherwise;  $u_D(x)=\max_{g\in G(x)} \{V(g)-V(x)\}$  for each $x\in D$ and $u_D(x)=-\infty$ otherwise.
Observe that since during flows there is no interaction among the oscillators, due to the definition of the flow map, $V$ is constant during flows, hence the growth of $V$ is bounded by $u_C, u_D$; see \cite{goebel2012hybrid}. Let, $U=\R^n\setminus \mathcal{X}$, we want to show that for each $x\in U$, one has that $u_D(x)\leq 0$. To this end, one needs to show that $\gamma(x)$ 
is non decreasing across jumps from $x\in U\cap D$. That is, one needs to prove that for each $x\in U\cap D, g\in G(x)$,  $\gamma(g)-\gamma(x)\geq 0$.
Let $\Theta=\Psi([\dstar,2\pi])$ and for each $x\in [0,2\pi]^n$ let $\alpha(x)\subset\Sone$ be the shortest arc containing $x$. Three different situations need to be analyzed; see \figurename~\ref{fig:TheoremMain}. Since $G(x)$ is single-valued for $x\in U\cap D$, we denote $g(x)$ as the unique element of $G(x)$. 

[Case 1 ]: $\alpha(x)\subset \Theta$; \figurename~\ref{fig:TheoremMain} (a).
Pick $x\in U\cap D$, and observe that necessarily for each $i\in\mathcal{V}$, $x_i\in[\dstar,\twopi]$ and that there exits a unique $k\in\mathcal{V}$ such that $x_i=2\pi$.  
Then due to Assumption~\ref{as:AsQ} and due to the definition of $G$, one has that $g_k(x)=0$ and for each $l\neq k\in\mathcal{V}$, $g_l(x)\leq 0$. Hence, it can be argued that $\gamma(g(x))\geq \gamma(x)$. 

[Case 2]: $\alpha(x)\supset \Theta$; \figurename~\ref{fig:TheoremMain} (b).
Pick $x\in U\cap D$, and observe that in this case $\gamma(g(x))=\gamma(x).$

[Case 3]: $\alpha(x)$ has an endpoint $h$ in $\Theta$; \figurename~\ref{fig:TheoremMain} (c).
Pick $x\in U\cap D$, and notice that either $\alpha(x)$ is the counterclockwise arc connecting $h$ to the other endpoint or $\alpha(x)\subset\Theta$, being the latter already discussed above.  
Then due to Assumption~\ref{as:AsQ} and due to the definition of $G$, one has that there exists $k\in\mathcal{V}$, such that $g_k(x)=0$. Moreover, for each $l\neq k\in\mathcal{V}$ such that $x_l\in [0,\dstar]$, one has $g_l(x)=0$, while for each $s\neq k\in\mathcal{V}$ such that $x_s\in (\dstar,2\pi]$, one has $g_s(x)<0$. Thus, it can be argued that $\gamma(g(x))\geq \gamma(x)$. 
Therefore, for each $x\in U$ one has that $u_D(x)\leq 0$. Hence, since $U$ is a neighborhood of $\mathcal{A}$, $V$ is continuous and positive definite with respect to $\mathcal{A}$ on $C\cup D$, $G(D)\subset C\cup D$, and $\mathcal{H}$ satisfies the hybrid basic conditions, by \cite[Theorem 8.8]{goebel2012hybrid} it follows that $\mathcal{A}$ is stable for $\mathcal{H}$. 
Now, we show that $\mathcal{A}$ is locally attractive for $\mathcal{H}$ implying that such a set is LAS. To this end, we make use of \cite[Theorem 8.8 (b)]{goebel2012hybrid}. In particular, we want to show that for each $r\in(0, \dstar]$ the largest weakly invariant\footnote{A set $S\subset\R^n$ is said to weakly invariant if for every $\xi\in S$, $\tau>0$, there exists at least one complete solution $\phi\in\mathcal{S}_{\mathcal{H}}(S)$ with $\range\phi\subset S$ (weak forward invariance) such that for some $(t^\star, j^\star)\in\dom\phi$ with $t^\star+j^\star\geq \tau$ one has $\phi(t^\star, j^\star)=\xi$ (weak backward invariance); see \cite{goebel2012hybrid} for further details on weak invariance in hybrid systems.} subset contained in $V^{-1}(r)\cap U\cap (u_C^{-1}(0)\cup (u_D^{-1}(0)\cap G(u_D^{-1}(0)))=V^{-1}(r)\cap U\cap C=V^{-1}(r)\cap \mathcal{U}$ is empty.
 To this end, it suffices to show that for every $r\in(0,\dstar]$, every $\phi\in\mathcal{S}_{\mathcal{H}}(V^{-1}(r)\cap \mathcal{U})$ leaves the set $V^{-1}(r)$. Pick any $\phi\in\mathcal{S}_{\mathcal{H}}(V^{-1}(r)\cap \mathcal{U})$, without loss of generality, suppose that $\phi(0,0)\in D$, and for each $(t,j)\in\dom\phi$, define $\psi(t,j)\colon\sort(\phi(t,j))$. Notice that, thanks to Assumption~\ref{as:AsQ},  if $\alpha(\xi)$ has an endpoint in $\Psi((\dstar,2\pi])$, then $V(\phi(0,1))<V(\phi(0,0))=r$. Hence, the only case to analyze is when the arc $\alpha(\phi(0,0))$ has an endpoint $h$ belonging to $\Sone\setminus\Psi((\dstar,2\pi])$.
Now observe that, since $r\in(0,\dstar]$, $\phi(0,0)\notin\mathcal{X}$, during flows the geodesic distance between oscillators' phase is unchanged and, thanks to Assumption~\ref{as:AsQ}, at jumps it never increases more than $\frac{\twopi}{n}$, it follows that for each $(t,j)\in\dom\phi$, there exists $i\in\mathcal{V}$ such that 
$\dist(\psi(t,j)_i(t,j),\psi(t,j)_{\nex(i)}(t,j))<\frac{\twopi}{n}$. 
Therefore, since from Assumption~\ref{as:AsQ}  at each jump all oscillators whose phase is in $(\dstar,2\pi)$ decrease their phases and the endpoint $h$ of $\alpha(\phi(t,j))$ by assumption does not belong to $\Psi((\dstar,2\pi))$,   
it can be easily observed that for each $(t,j)\in\dom\phi$, there exist $(t_s, j_s), (t_s, j_s+1)\in \dom\phi$ such that 
$\alpha(\phi(t_s, j_s))$ has an endpoint in $\Psi((\dstar,2\pi))$, implying that $\gamma(\phi(t_s, j_s+1))>\gamma(\phi(t_s, j_s))=r$.
 Hence, according to \cite[Theorem 8.8]{goebel2012hybrid}, $\mathcal{A}$ is locally asymptotically stable for $\mathcal{H}$. To conclude the proof, we need to show that the basin of attraction of $\mathcal{A}$ includes $\mathcal{U}$. Indeed, in that case one has 
 $\mathcal{B}_{\mathcal{A}}=(\R^n\setminus(\mathcal{U}\cup\mathcal{X}))\cup \mathcal{U}=\R^n\setminus\mathcal{X}$, which guarantees that $\mathcal{A}$ is almost globally asymptotically stable since $\R^n\setminus\mathcal{B}_{\mathcal{A}}=\mathcal{X}$ and \footnote{The set $\mathcal{X}$ is the union of a finite number of isolated points and affine subspaces of $\R^n$ of dimension $n-1$, \ie, sets of Lebesgue measure zero. Therefore, by countable subadditivity of the Lebesgue measure, one has $m(\mathcal{X})=0$; see \cite{schilling2005measures}.} $m(\mathcal{X})=0$.
To this end, from \cite[Theorem 8.2]{goebel2012hybrid}, one that each $\phi\in\mathcal{S}_{\mathcal{H}}(\mathcal{U})$ converges\footnote{Maximal solutions to $\mathcal{H}$ are complete, bounded (precompact) and do not leave $C$ due to $C\cup D\cup G(D)=C$ compact.}, for some $r\in V(C)$,  to the largest and nonempty weakly invariant subset of $V^{-1}(r)\cap C$. Then since we shown that for each $r\in(0,\dstar]$,  any $\phi\in\mathcal{S}_{\mathcal{H}}(\mathcal{U})$ leaves  $V^{-1}(r)$, if follows that each $\phi\in\mathcal{S}_{\mathcal{H}}(\mathcal{U})$ approaches $\mathcal{A}$ and this finishes the proof.
\end{proof}
\begin{figure}[h!]
\centering
 \subfigure[]
   {\includegraphics[scale=0.32]{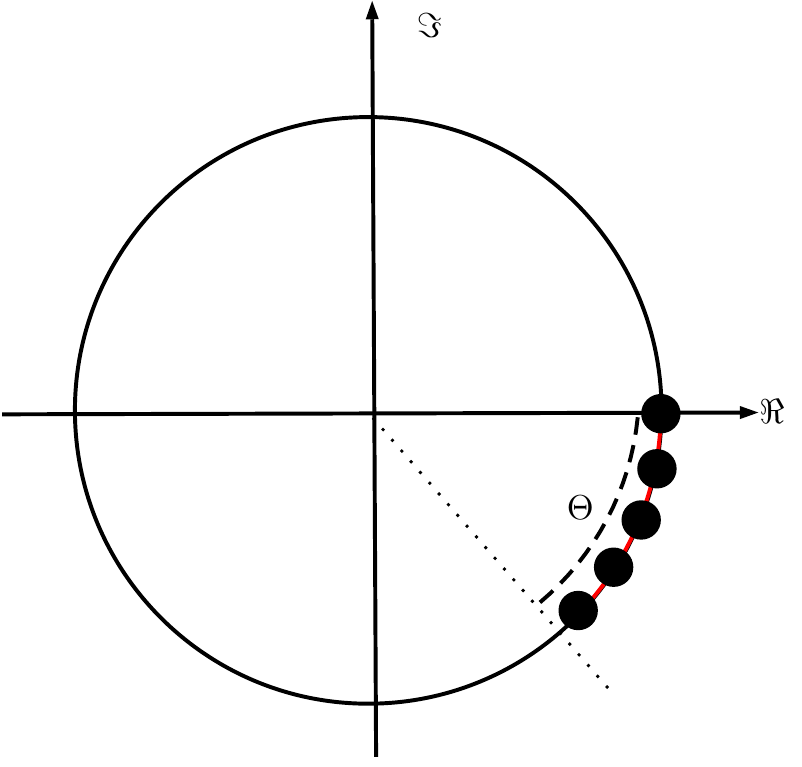}}\,
 \subfigure[]
   {\includegraphics[scale=0.32]{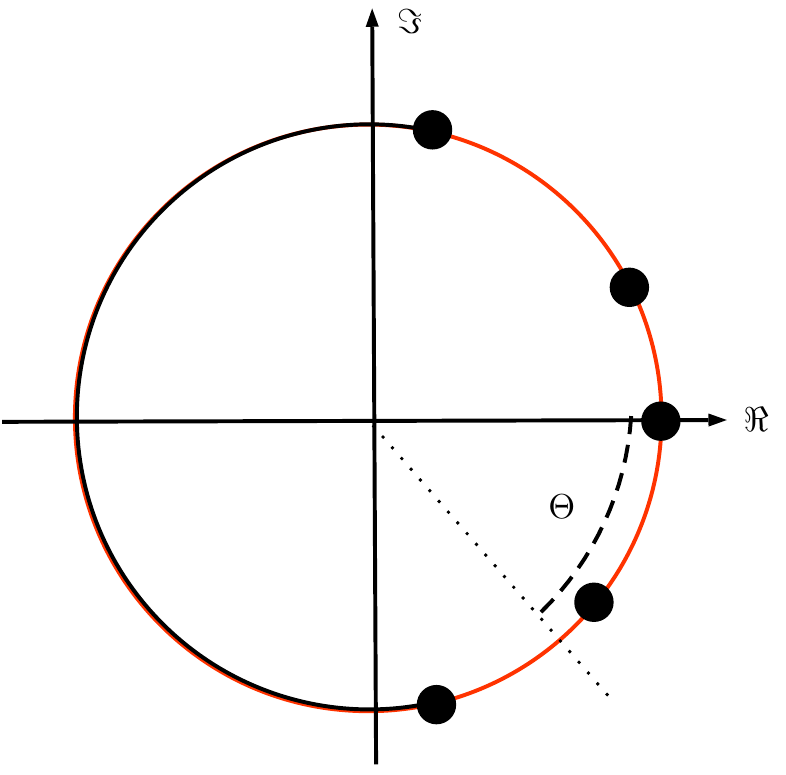}}\,
 \subfigure[]
   {\includegraphics[scale=0.32]{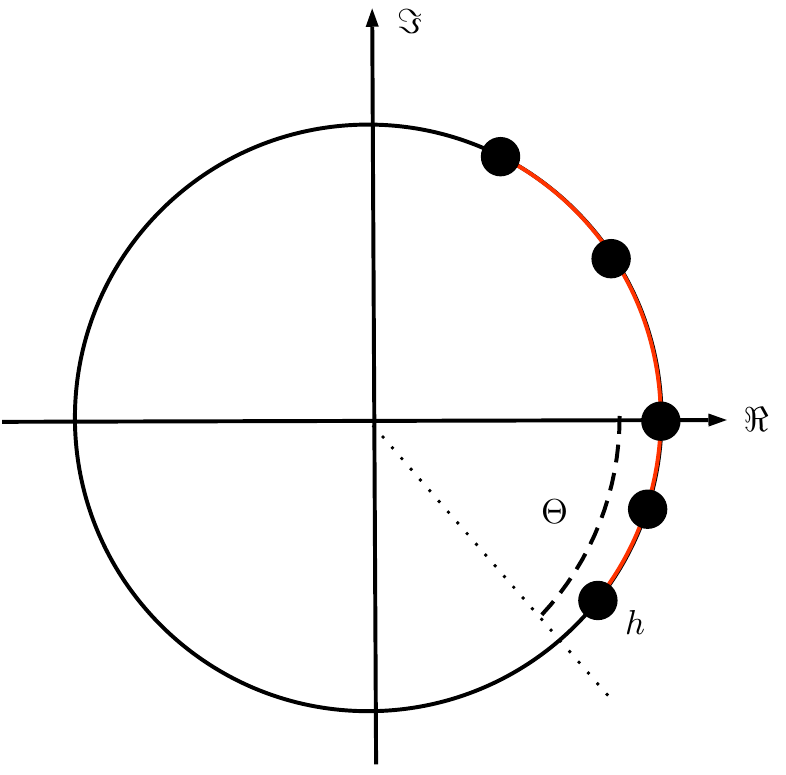}}
\caption{{\small The different cases discussed in the proof of Theorem~\ref{theorem:Main}. (a) Case 1, (b) Case 2, (c) Case 3, $\alpha(x)$ (red), $\Theta$ (dashed-black).}}
 \label{fig:TheoremMain}
 \end{figure}
 \begin{remark}
The above result makes use of the notion of the shortest containing arc to ensure almost global asymptotic stability of the set $\mathcal{A}$.  
Nonetheless, due to the geometric characterization of the set $\mathcal{A}$ illustrated in Remark~\ref{rem:line}, it may be tempting to assume that the Lyapunov analysis proposed in \cite{philips2015} building on the function $x\mapsto \widetilde{V}(x)=\vert x \vert_{\mathcal{L}}$ (where the set $\mathcal{L}$ is defined in Remark~\eqref{rem:line}) 
can be adapted to deal with the setting considered in this paper. To this end, it is worthwhile to observe that in \cite{philips2015}, to assess asymptotic stability of the desynchronization set via the adoption of the proposed distance-like function, the authors make use of the inherent relationships between the desynchronization set and the considered phase response function. However, due to different phase response function and set to be stabilized, such relationships do not hold in our setting. Thus, the arguments in \cite{philips2015} cannot be directly employed to assess asymptotic stability of the set 
$\mathcal{A}$ via the use of the proposed distance-like function.
In fact, numerical experiments show that such a distance-like function, although approaching asymptotically zero,  may increase along the solutions of the hybrid system $\mathcal{H}$ when the phase response function is designed according to the prescriptions given by Assumption~\ref{as:As1} and Assumption~\ref{as:AsQ}, making such a function unsuitable as a Lyapunov candidate in our setting. This aspect is made evident in Section IV through numerical simulations.
 \end{remark}
\subsection{Robustness analysis}
In real-world settings, due to parameter variations or constructive imperfections, assuming identical natural frequency $\omega$ for each oscillator is unrealistic. Therefore, having an insight on the behavior of the hybrid system $\mathcal{H}$ in the presence of small perturbations on the oscillators' natural frequency is an interesting aspect. To delve into this issue, consider the following hybrid model of the network of PCOs with frequency perturbations
\begin{equation}
\label{eq:Pert}
\mathcal{H}_d\left\{
\begin{array}{ll}
\begin{array}{rcl}
\dot{x}& =& \omega\1_n+d\\
\end{array}& x\in C\\
\begin{array}{rcl}
x^{+} &\in&G(x)\\
\end{array}&x\in D
\end{array}
\right.
\end{equation}
where $d\in\R^n$ represents the perturbation affecting the flow dynamics. To formally describe the deviation of $\mathcal{H}_d$ from the nominal behavior captured by $\mathcal{H}$, we make use of the notion of  $(\tau,\varepsilon)$-closeness of hybrid arcs given in Section~\ref{sec:hybrid}. In particular, by relying on the regularity of the data of 
$\mathcal{H}$, we can establish the following result.\footnote{\textcolor{blue}{In the published version, there is a typo in the proof of Proposition 3. Namely, $\mathcal{H}_{\rho}$ should be $\mathcal{H}^\prime_{\rho}$ to avoid confusion with $\mathcal{H}_d$ in \eqref{eq:Pert}.}}
\begin{proposition}
For each $\varepsilon>0$ and $\tau\geq 0$, there exists $\delta>0$ such that for each measurable function $d\colon\R_{\geq 0}\rightarrow\delta\mathbb{B}$ and for each solution\footnote{Given a measurable function $d\colon\R_{\geq 0}\rightarrow\R^n$, $\phi_d$ is a solution to $\mathcal{H}_d$ if $(\phi, d')$ with $d'(t,j)=d(t)$ for each $(t,j)\in \dom\phi$ is a solution pair to $\mathcal{H}_d$; see \cite{cai2009characterizations} for further details.} $\phi_d$ to \eqref{eq:Pert} from $C$, there exists a solution $\phi$ to $\mathcal{H}$ from $C$ such that $\phi$ and $\phi_d$ are $(\tau,\varepsilon)$-close.
\end{proposition}
\begin{proof}
Let $\rho$ be a positive real scalar, define \textcolor{blue}{$\mathcal{H}^\prime_{\rho}=(C, F_\rho,D,G)$}, where for each $x\in C$ $F_\rho(x)=f(x)+\rho\mathbb{B}$.
Since $\mathcal{H}$ satisfies the hybrid basic conditions and is pre-forward complete\footnote{Given a set $S$, a hybrid system $\mathcal{H}$ is said to be pre-forward complete from $S$ if every $\phi\in\mathcal{S}_{\mathcal{H}}(S)$ is either complete or bounded; see \cite[Definition 6.12]{goebel2012hybrid}.} from the (compact) set $C$, from \cite[Proposition 6.34]{goebel2012hybrid}  it follows that for every $\varepsilon>0$ and $\tau\geq 0$ there exists $\delta>0$ such that for every solution $\phi_\delta$ to \textcolor{blue}{$\mathcal{H}^\prime_{\delta}$} there exists a solution $\phi$ to $\mathcal{H}$ such that $\phi$ and 
$\phi_\delta$ are $(\tau,\varepsilon)$-close. To conclude, if suffices to notice that whenever $d(t)\in\delta\mathbb{B}$ for each $t\geq 0$, solutions to $\mathcal{H}_d$ are solutions to $\textcolor{blue}{\mathcal{H}^\prime_{\delta}}$ and this concludes the proof.   
\end{proof}
 Essentially, the above result states that in the presence of small perturbations on the natural frequency, the evolution of the perturbed network of PCOs does not differ too much from the one of the unperturbed network. Therefore, one can expect that the convergence towards the splay state is practically preserved in the case of (small) frequency perturbations.      
\section{Numerical Example}
\subsection{Nominal behavior}
In this example in which $n=3$, we want to show both the effectiveness of the proposed methodology and the differences between the dynamics of $\mathcal{H}$ and the results given in \cite{philips2015}. We use the following phase response function
$$
Q(z)=\begin{cases}
0&\,\,\mbox{if}\,\,z\in[0,\frac{4\pi}{3}]\\
-\frac{7}{10}(z-\frac{4\pi}{3})&\,\,\mbox{if}\,\, z\in[\frac{4\pi}{3}, \twopi]
\end{cases}
$$  
\figurename~\ref{fig:plots2} shows the evolution of a solution $\phi$ to $\mathcal{H}$ starting from an initial condition 
$\phi(0,0)\notin\mathcal{A}$. As expected, as $t+j$ goes to infinity, $\phi$ approaches the splay state configuration and the Lyapunov function $V$ is constant during flows, nonincreasing across jumps, and approaches $0$. 
\begin{figure}[h]
\begin{center}
\pstool[scale=0.4,trim=15mm 4mm  15mm 8mm, clip=true]{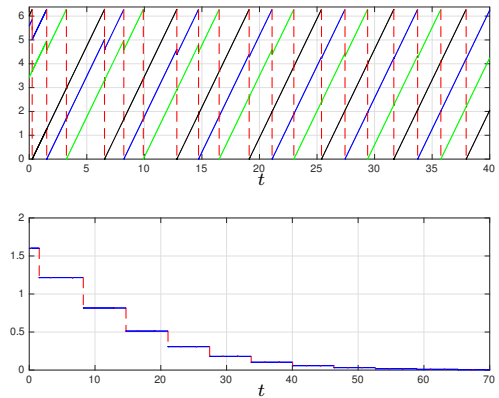}
{\scriptsize
\psfrag{t}{$t$}
\psfrag{phi}{}
\psfrag{V(phi)}{}
}
\end{center}
\caption{Evolution of the solution $\phi$ (top) to $\mathcal{H}$ with $\phi(0,0)=(5.5977, 6.0274, 3.4383)$ and of $V\circ\phi$ (bottom) both projected onto ordinary time.}
\label{fig:plots2}
\end{figure}
Now, as mentioned earlier, we show that the Lypunov like function proposed in \cite{philips2015}, \ie, $x\mapsto \widetilde{V}(x)=\vert x \vert_{\mathcal{L}}$, cannot be employed in the setting analyzed in this paper. To this end, in \figurename~\ref{fig:dist} we report the evolution of $\widetilde{V}$ along the solution $\phi$ to $\mathcal{H}$ analyzed above. Simulations show that, although $\widetilde{V}$ approaches zero as $t+j$ goes to infinity, the considered function increases across some jumps and this makes it unsuitable as a Lyapunov candidate to assess asymptotic stability of the set $\mathcal{A}$ for the hybrid system $\mathcal{H}$.
\begin{figure}[h]
\centering
\pstool[scale=0.4,trim=15mm 1mm  15mm 5mm, clip=true]{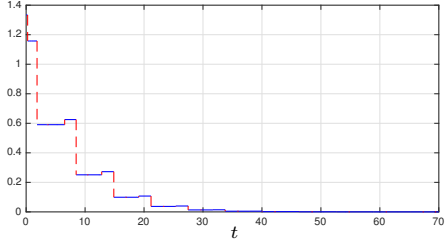}
{\scriptsize
\psfrag{t}{$t$}
\psfrag{dist}{$\widetilde{V}\circ\phi$}
}
\caption{Evolution of the function $\widetilde{V}\circ\phi$ with $\phi(0,0)=(5.5977, 6.0274, 3.4383)$ projected onto ordinary time.}
\label{fig:dist}
\end{figure}
\subsection{Perturbed behavior}
Now, we want to analyze the effect of  small frequency perturbation on the dynamics of $\mathcal{H}$. In particular,  to study the effect of uncorrelated periodic frequency perturbations, we select $d(t)=\epsilon (\sin(0.5 t),\sin(0.5 t+2/3\pi), \sin(0.5 t+4/3\pi))$, where $\epsilon>0$ represents the maximum amplitude of the frequency variation. \figurename~\ref{fig:plotsPert} reports the evolution of $\mathcal{H}_d$ and of $\mathcal{H}$ for $\epsilon$ equal to $0.05$ and $0.03$, repsectively. As expected, the evolution of $\mathcal{H}_d$ does not differ too much from the evolution of $\mathcal{H}$. \figurename~\ref{fig:plotsV} shows the evolution of the function $V$ along, respectively, the evolution of $\mathcal{H}$ and of $\mathcal{H}_d$ for the two considered values of $\epsilon$. Simulations show that the function $V$ in the presence of small frequency perturbations approaches a ball containing zero. In that simulation, one can remark that the larger the amplitude of the perturbation, the larger the radius of the ball approached by $V$.
\begin{figure}[h]
\centering
\pstool[scale=0.5, trim=15mm 4mm  12mm 7mm, clip=true]{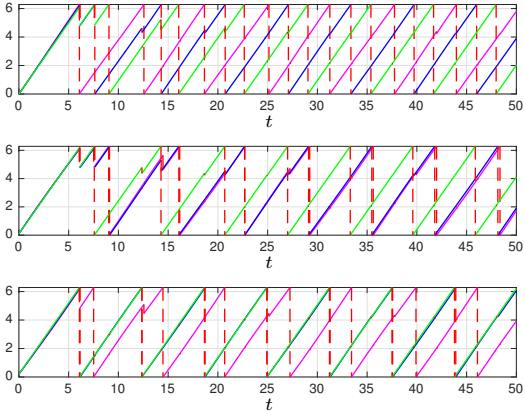}
{\scriptsize
\psfrag{t}{$t$}}
\caption{Evolution of the solution $\phi$ (green) to $\mathcal{H}$ and $\xi$ to $\mathcal{H}_d$ ($\epsilon=0.05$ magenta), ($\epsilon=0.03$ blue) with $\phi(0,0)=\xi(0,0)=(0, 0.1,0.2)$ projected onto ordinary time.}
\label{fig:plotsPert}
\end{figure}
\begin{figure}[h!]
\centering
\pstool[scale=0.4, trim=15mm 1mm  12mm 2mm, clip=true]{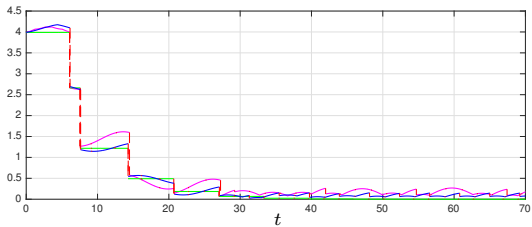}
{\scriptsize
\psfrag{t}{$t$}} 
\caption{Evolution of the function $V\circ\phi$ (green) and of $V\circ\xi$ ($\epsilon=0.05$ magenta), ($\epsilon=0.03$ blue) both projected onto ordinary time.}
\label{fig:plotsV}
\end{figure}
\section{Conclusion}
In this paper, we studied splay state almost global asymptotic stabilization in pulse coupled oscillators. The considered problem was turned into the stabilization problem of a compact set $\mathcal{A}$ wherein oscillators' phases are evenly distributed on the unit circle.  Sufficient conditions on the phase response function to guarantee almost global asymptotic stability of set $\mathcal{A}$ were provided. In particular, almost global asymptotic stability of the set $\mathcal{A}$ is assessed via the use of a novel Lyapunov function along with the invariance principle for hybrid systems in \cite{sanfelice2007invariance}. Moreover, the proposed approach was shown to be robust to small frequency perturbations naturally present in practice. 
The results presented in this technical note are promising and the framework adopted quite flexible to envision interesting extensions of the results presented here. Among them, we mention the extension to more complex network topologies.
\bibliographystyle{plain}
\bibliography{Biblio}
\end{document}